\documentclass
{eptcs}
\providecommand{\event}{AFL 2014} 

\usepackage{amsmath}
\usepackage{amsfonts}
\usepackage{amssymb}
\usepackage{graphicx}
\usepackage{amsthm}

\usepackage{tabularx}
\usepackage{stmaryrd}

\usepackage{hyperref}

\newcommand{\eps}{\varepsilon}

\newcommand{\nsc}{\mathrm{nsc}}

\newcommand{\isc}{\mathrm{isc}}

\renewcommand{\sc}{\mathrm{sc}}

\newtheorem{theorem}{Theorem}
\newtheorem{lemma}[theorem]{Lemma}
\newtheorem{proposition}[theorem]{Proposition}

\newenvironment{definition}[1][Definition]{\begin{trivlist}
\item[\hskip \labelsep {\bfseries #1}]}{\end{trivlist}}
\newenvironment{example}[1][Example]{\begin{trivlist}
\item[\hskip \labelsep {\bfseries #1}]}{\end{trivlist}}

\title{Operations on Automata with All States Final}

\author{Krist\'{i}na \v{C}evorov\'a
\institute{Mathematical Institute,\\
              Slovak Academy of Sciences,\\
              Bratislava, Slovakia}
\email{cevorova@mat.savba.sk}
\and
Galina Jir\'askov\'a\thanks{Research supported by grant APVV-0035-10.}
\quad\quad Peter Mlyn\'ar\v cik \quad\quad Mat\'u\v s Palmovsk\'y
\institute{Mathematical Institute,\\
              Slovak Academy of Sciences,\\
              Ko\v{si}ce, Slovakia}
\email{ jiraskov@saske.sk \quad mlynarcik1972@gmail.com \quad matp93@gmail.com}
\and
Juraj \v{S}ebej\thanks{Research supported by grant  VEGA 1/0479/12.}
\institute{Institute of Computer Science,\\
              \v Saf\'arik University,\\
              Ko\v{s}ice, Slovakia}
\email{juraj.sebej@gmail.com}
}

\def\titlerunning{Operations on Automata with All States Final}
\def\authorrunning{K. \v Cevorov\'a, G. Jir\'askov\'a, P. Mlyn\'ar\v cik, 
M. Palmovsk\'y, J. \v{S}ebej}

\begin{document}
\maketitle

\begin{abstract}
We study the complexity of basic regular operations on languages
represented by incomplete deterministic or nondeterministic automata,
in which all states are final. 
Such languages are known to be prefix-closed.
We get  tight bounds on both incomplete and nondeterministic
state complexity of complement, intersection, union, 
concatenation,  star, and reversal on prefix-closed languages. 
\end{abstract}

\section{Introduction}

A language $L$ is prefix-closed
if $w\in L$ implies that every prefix of $w$ is in $L$.
It is known that a regular language is prefix-closed
if and only if it is accepted by a nondeterministic finite automaton (NFA)
with all states final \cite{krs09}.
In the minimal incomplete deterministic finite automaton (DFA)
for a prefix-closed language, all the states are final as well.

The authors of \cite{krs09} examined several questions
concerning NFAs with all states final.
They proved that the inequivalence problem for NFAs with all states final 
is PSPACE-complete in the binary case,
but polynomially solvable in the unary case.
Next, they showed that minimizing a binary NFA with all states final is 
PSPACE-hard, and that deciding whether a given NFA accepts a language that is not prefix-closed is PSPACE-complete, while the same problem for DFAs
can be solved in polynomial time.
The NFA-to-DFA conversion and complementation of NFAs with all states final
 have been also considered in \cite{krs09},
and the tight bound $2^n$ for the first problem, and the lower bound $2^{n-1}$
for the second one have been obtained.

The quotient complexity of prefix-closed languages
has been studied  in \cite{bjz14}.
The quotient of a language $L$ by the string $w$
is the set $L_w=\{x\mid wx\in L\}$.
The quotient complexity of a language $L$, $\kappa(L)$,
is the number of distinct quotients of $L$.
Quotient complexity  is defined for any language, 
and it is finite if and only if the language is regular.
The quotient automaton of a regular language $L$
is the DFA $(\{L_w\mid w\in\Sigma^*\},\Sigma,\cdot,L_\eps,F)$,
where $L_w\cdot a = L_{wa}$, and a quotient $L_w$ is final
if it contains the empty string. 
The quotient automaton of $L$ is a minimal complete DFA for $L$,
so quotient complexity is the same as the state complexity of $L$
which is defined as the number of states in the minimal DFA for $L$.
In \cite{bjz14}, the tight bounds on the quotient complexity
of basic regular operation have been obtained,
and to prove upper bounds, the properties of quotients have been used
rather than automata constructions. 

Automata with all states final 
represent  systems,  for example,  production lines,
and their intersection or parallel composition
represents the composition of these systems \cite{ma13}. 
A question that arises here is, 
whether the complexity of intersection
of automata with all states final
is the same as in the general case of arbitrary DFAs or NFAs.
At the first glance, it seems that this complexity could be smaller.
Our first result shows that this is not the case.
We show that both incomplete and nondeterministic state complexity
of intersection on prefix-closed languages is given by the function $m n$,
which is the same as in the general case of regular languages.

In the deterministic case, to have all the states final, 
we have to consider incomplete deterministic automata
because otherwise, the complete automaton 
with all states final would accept
the language consisting of all the strings over an input alphabet.
Notice that the model of incomplete deterministic automata
has been considered already by Maslov \cite{ma70}.
The same model has been used in the study of the complexity
of the shuffle operation \cite{csy02}; 
here, the complexity on complete DFAs is not known yet.

We next study the complexity of complement, union,
concatenation, square, star, and reversal on languages
represented by incomplete DFAs or NFAs with all states final.
We get  tight bounds in both nondeterministic
and incomplete deterministic cases.
In the nondeterministic case, all the bounds are the same 
as in the general case of regular languages,
except for the bound for star that is $n$ instead of  $n+1$.
However, to prove the tightness of these bounds,
we usually  use  larger alphabets 
than in the general case of regular languages where all the upper bounds
can be  met by binary languages \cite{hk03,ji05}.

To get lower bounds, we use a fooling-set lower-bound method
\cite{auy83,bi92,bi93,gs96,hr97}.
In the case of union and reversal,
the method does not work since it provides a lower bound
on the size of NFAs with multiple initial states.
Since the nondeterministic state complexity of a regular language
is defined using a model of NFAs with a single initial state \cite{hk03},
we have to use a modified fooling-set technique to get the
tight bounds $m+n+1$ and $n+1$ for union and reversal,
respectively.

In the case of incomplete deterministic finite automata,
the tight bounds for complement, union,
concatenation,  star, and reversal are
$n+1, mn+m+n, m\cdot2^{n-1}+2^n-1$, 
$2^{n-1}$, and
$2^n-1$,
respectively.
To define worst-case examples,
we use a binary alphabet for union, star, and reversal,
and a ternary alphabet for concatenation.

The paper is organized as follows.
In the next section, we give some basic definitions
and preliminary results.
In Sections~\ref{se:co} and \ref{se:bo}, we  study boolean operations.
Concatenation is discussed in Section~\ref{se:cs},
and star and reversal in Section~\ref{se:sr}.
The last section contains some concluding remarks.

\section{Preliminaries}

In this section, we recall some basic definitions
and preliminary results.
For details and all unexplained notions,
the reader may refer to \cite{si97}.

A \emph{nondeterministic finite automaton} (NFA)
is a quintuple $A=(Q,\Sigma,\delta,I,F)$,
where $Q$ is a finite set of states,
$\Sigma$ is a finite alphabet,
$\delta\colon Q\times\Sigma\to 2^Q$ is the transition function
which is extended to the domain $2^Q\times\Sigma^*$
in the natural way,
$I\subseteq Q$ is the set of initial states, and
$F\subseteq Q$ is the set of final states.
The language accepted by $A$ is the set
$L(A)=\{w\in\Sigma^*\mid \delta(I,w)\cap F \neq \emptyset\}$.

The \emph{nondeterministic state complexity}
 of a regular language $L$, $\nsc(L)$,
is the smallest number of states in any NFA  
with a \emph{single initial state}  recognizing $L$.

An NFA $A$ is \emph{incomplete deterministic} (DFA)
if $|I|=1$ and $|\delta(q,a)|\le1$
for each $q$ in $Q$ and each $a$ in $\Sigma$.
In such a case, we write  $\delta(q,a)=q'$ 
instead of $\delta(q,a)=\{q'\}$.
A non-final state $q$ of a DFA is called a \emph{dead} state
if $\delta(q,a)=q$ for each symbol $a$ in $\Sigma$.

The \emph{incomplete state complexity} of a regular language $L$,
$\isc(L)$,
is the smallest number of states in any incomplete DFA recognizing $L$.
An incomplete DFA is minimal (with respect to the number of states)
if it does not have any dead state,
all its states are reachable,
and no two distinct states are equivalent.

Every NFA $A=(Q,\Sigma,\delta,I,F)$ can be converted
to an equivalent DFA $A'=(2^Q,\Sigma,\cdot,I,F')$,
where $R\cdot a =\delta(R,a)$ and $F'=\{R\in 2^Q\mid R\cap F\neq\emptyset\}$.
The DFA $A'$ is called 
the \emph{subset automaton} of the NFA~$A$.
The subset automaton need not be minimal
since some of its states may be unreachable or equivalent.
However, if for each state $q$ of an NFA $A$,
there exists a string $w_q$ that is accepted by $A$
only from the state~$q$,
then the subset automaton of the NFA $A$
does not have equivalent states 
since if two  subsets of the subset automaton
differ in a state $q$, 
then they are distinguishable by $w_q$.

To prove the minimality of NFAs,
we use a fooling set lower-bound technique, 
see~\cite{auy83,bi92,bi93,gs96,hr97}.

\begin{definition}\label{def:fool}
 A set of pairs of strings $\{(x_1,y_1),(x_2,y_2),\ldots,(x_n,y_n)\}$
 is called a \emph{fooling set} for a language $L$ if
 for all $i,j$ in $\{1,2,\ldots,n\}$, the following two conditions hold:\\
 \hglue10pt
 \textbf{\emph{(F1)}} $x_iy_i\in L$, and \\ 
 \hglue10pt
 \textbf{\emph{(F2)}} if $i\neq j$, then  $x_iy_j\notin L$ or $x_jy_i \notin L$. 
\end{definition}

It is well known that the size of a fooling set for a regular language 
provides a lower bound on the number of states 
in any NFA (with multiple initial states)
for the language.
The argument is simple.
Fix the accepting computations of any NFA on strings $x_i y_i$ and $x_j y_j$.
Then, the states on these computations reached after reading $x_i$ and $x_j$
must be distinct, otherwise the NFA
accepts both $x_i y_j$ and $x_j y_i$ for two distinct pairs.
Hence we get the following observation.

\begin{lemma}[\cite{bi93,gs96,hr97}]
\label{le:fool}
  Let $\mathcal{F}$ be a fooling set for a language $L$.
 Then every NFA (with multiple initial states) for the language $L$
 has at least $|\mathcal{F}|$ states.
\qed
\end{lemma}

The next lemma  shows that sometimes,
if we insist on having a single initial state in an NFA,
one more state is necessary.
It can be used in the case of union, reversal, cyclic shift \cite{jo08},
and AFA-to-NFA conversion \cite{ji12}.
In each of these cases,  NFAs
with a single initial state
require one more state than NFAs
with multiple initial states. 
For the sake of completeness,
we recall the proof of the lemma here.

\begin{lemma}[\cite{jm11}]\label{le:fool1}
 Let $\mathcal{A}$ and $\mathcal{B}$ be sets of pairs of strings
 and let $u$ and $v$ be two strings such that
 $\mathcal{A}\cup\mathcal{B}$,
 $\mathcal{A}\cup\{(\eps,u)\}$, and
 $\mathcal{B}\cup\{(\eps,v)\}$ are fooling sets for a  language $L$.
 Then every NFA with a single initial state for the language $L$
 has at least $|\mathcal{A}|+|\mathcal{B}|+1$ states.
\end{lemma}

\begin{proof}
 Consider an NFA for a language $L$, and let
 $\mathcal{A}=\{(x_i,y_i)\mid i=1,2,\ldots,m\}$ and
 $\mathcal{B}=\{(x_{m+j},y_{m+j})\mid j=1,2,\ldots,n\}$.
 Since the strings $x_k y_k$ are in $L$,
 we fix an accepting computation of the NFA  on each string $x_k y_k$.
 Let $p_k$ be the state on this computation
 that is reached after reading $x_k$.
 Since $\mathcal{A}\cup\mathcal{B}$ is a fooling set for $L$,
 the states $p_1$, $p_2$,~\ldots, $p_{m+n}$ are pairwise distinct.
 Since $\mathcal{A}\cup\{(\eps,u)\}$ is a fooling set,
 the  initial state is distinct from all the states $p_1$, $p_2$, \ldots, $p_m$.
 Since $\mathcal{B}\cup\{(\eps,v)\}$ is a fooling set,
 the (single) initial state is also distinct
 from all the states $p_{m+1}$, $p_{m+2}$, \ldots, $p_{m+n}$.
 Thus the NFA has at least $m+n+1$ states.
\end{proof}

\begin{example}\label{ex:fooling}
  Let $K=(a^3)^*$ and $L=(b^3)^*$.
 Then $\nsc(K)=3$ and $\nsc(L)=3$,
 and the language $K\cup L$ is accepted
 by a 6-state NFA with two initial states.
 Therefore, we cannot expect that we will be able to find a fooling set
 for $K\cup L$ of size $7$.
 However, every NFA with a \emph{single} initial state for the language $K\cup L$
 requires at least $7$ states since Lemma~\ref{le:fool1}
 is satisfied for the language $K\cup L$
 with 
 \begin{align*}
   \mathcal{A} &=\{(a,a^2),(a^2,a),(a^3,a^3)\},\\
   \mathcal{B} &=\{(b,b^2),(b^2,b),(b^3,b^3)\},\\
   u                &=b^3, \text{ and} \\
   v                &=a^3.
 \end{align*}
\end{example}

If $w=uv$ for strings $u$ and $v$, then $u$ is a \emph{prefix} of $w$.
A language $L$ is \emph{prefix-closed}
if $w\in L$ implies that every prefix of $w$ is in $L$.
The following observations are easy to prove.

\begin{proposition}[\cite{krs09}]
 A regular language is prefix-closed if and only if
 it is accepted by some NFA with all states final.
\qed
\end{proposition}

\begin{proposition}
 Let $A$ be a minimal incomplete DFA for a language $L$.
 Then the language $L$ is prefix-closed 
 if and only if all the states of the DFA $A$ are final.
\qed
\end{proposition}

\section{Complementation}
\label{se:co}

If $L$ is a language over an alphabet $\Sigma$,
then the complement of $L$ is the language $L^c=\Sigma^*\setminus L$.
If $L$ is accepted by a minimal complete DFA $A$,
then we can get a minimal DFA for $L^c$
from the DFA $A$ by interchanging the final and non-final states.
In the case of incomplete DFAs, we first have to add a dead state,
that is, a non-final state which goes to itself on each input,
and let all the undefined transitions go to the dead state.
After that, we can interchange the final and non-final states
to get a (complete)  DFA for the complement.
This gives the following result.

\begin{theorem}
 Let $n\ge1$.
 Let $L$ be a prefix-closed regular language over an alphabet $\Sigma$
 with $\isc(L)=n$.
 Then $\isc(L^c)\le n+1$, and the bound is tight if $|\Sigma|\ge1$.
\end{theorem}

\begin{proof}
 For tightness, we can consider the unary prefix-closed language
 $\{a^i\mid 0\le i\le n-1\}$. 
\end{proof}

If a language $L$ is represented by an $n$-state NFA,
then we first construct  the corresponding subset automaton,
and then interchange the final and non-final states to get a DFA for the language $L^c$ of at most $2^n$ states.
This upper bound on the nondeterministic state complexity of complement 
on regular languages
is know to be tight in the binary case \cite{ji05}.

For prefix-closed languages, we get the same bound,
however, to prove tightness, we use a ternary alphabet.
Whether or not the bound $2^n$ can be met by a binary language
remains open.

\begin{theorem}
 Let $n\ge2$.
 Let $L$ be a prefix-closed regular language
 over an alphabet  $\Sigma$ with $\nsc(L)=n$.
 Then $\nsc(L^c)\le 2^n$,
 and the bound is tight if $|\Sigma|\ge3$.
\end{theorem}

\begin{proof}
 The upper bound is the same as
 in the general case of regular languages \cite{hk03}.
 To prove tightness,
 consider the language $L$
 accepted by the NFA $N$ shown in Figure~\ref{fig:complement},
 in which state $n$ goes to the empty set on both $a$ and $b$,
 and to $\{1\}$ on $c$.
 Each other state $i$ goes to $\{i+1\}$ on both $a$ and $c$, and 
 to $\{1,i+1\}$ on~$b$.
 Our~aim is to describe a fooling set 
 $\mathcal{F}=\{(x_S,y_S)\mid S \subseteq \{1,2,\ldots,n\}\}$
 of size $2^n$ for $L^c$.

 \begin{figure}[b]\label{-----fi1}
 \centerline{\includegraphics[scale=.40]{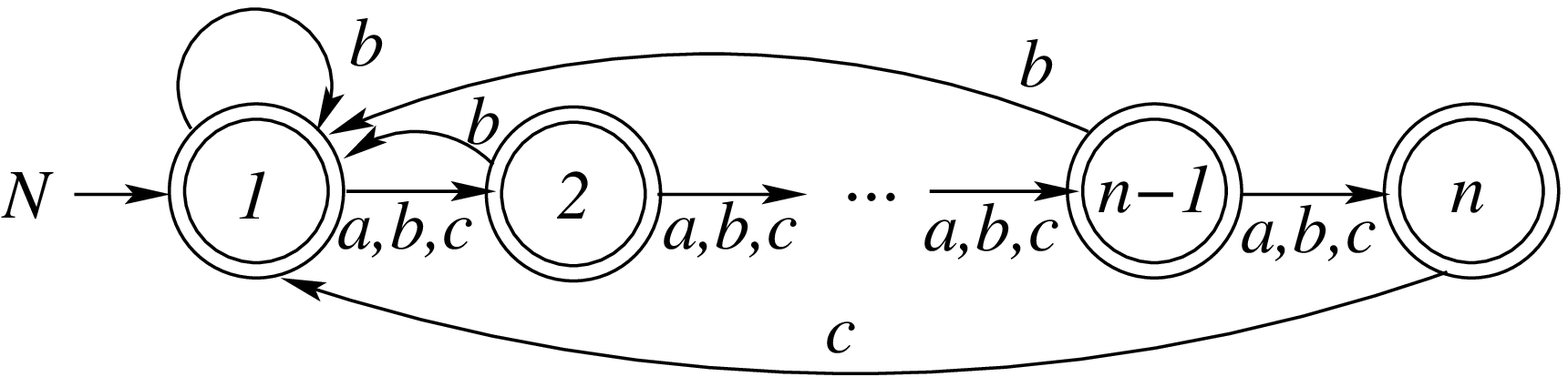}}
 \caption{The NFA $N$ of a prefix-closed language $L$ with $\nsc(L^c)=2^n$.}
 \label{fig:complement}
 \end{figure}

 First, let us show that each subset of $\{1,2,\ldots,n\}$
 is reachable in the subset automaton of the NFA $N$.
 The initial state is $\{1\}$,
 and each singleton set $\{i\}$ is reached from $\{1\}$
 by $a^{i-1}$.
 The empty set is reached from $\{n\}$ by $a$.
 The set $\{i_1,i_2,\ldots,i_k\}$ of size $k$,
 where  $2\le k\le n$ and $1\le i_1 < i_2 <\cdots < i_k\le n$,
 is reached from the set  $\{i_2-i_1,\ldots,i_k-i_1\}$ of size $k-1$
 by the string $b a^{i_1-1}$.
 This proves reachability by induction.
 Now, define $x_S$ as the string, by which the initial state $1$
 of the NFA $N$ goes to the set $S$.

 Next, for a subset $S$ of $\{1,2,\ldots,n\}$,  define the string $y_S$ 
 as the string $y_S=y_0y_1\cdots y_{n-1}$
 of length $n$, where
 $$
     y_i = \begin{cases}
             a, &\text{if $n-i\in S$,}\\
             c, &\text{if $n-i\notin S$.}
       \end{cases}
 $$
 We claim that the string $y_S$ is rejected by the NFA $N$
 from each state in $S$ and accepted from each state that is not in $S$.
 Indeed, if $i$ is a state in $S$,
 then $y_{n-i}=a$ and $y_S=u a v$ with $u=y_0y_1\cdots y_{n-i-1}$
 and $v=y_{n-i+1}y_{n-i+2}\cdots y_{n-1}$.
 Hence $|u|=n-i$,
 which means that the state $i$ goes to $\{n\}$ by $u$
 since both $a$ and $c$ move each state $q$ to state $q+1$.
 However, in state $n$ the NFA $N$ cannot read $a$,
 and therefore the string $y_S=uav$ is rejected from $i$.
 On the other hand,
 if $i\notin S$, then $y_{n-i}=c$,
 and the string $y_S=u c v$ with $|u|=n-i$ and $|v|=i-1$
 is accepted from $i$ through the computation
 $i\xrightarrow{u} n \xrightarrow{c} 1 \xrightarrow{v} i$.

 Now, we are ready to prove that the set of pairs of strings
 $ \mathcal{F} = \{(x_S,y_S) \mid S\subseteq\{1,2,\ldots,n\}\}$
 is a fooling set for the language $L^c$.

 (F1) By $x_S$, the initial state 1 goes to the set $S$.
 The string $y_S$ is rejected by $N$ from each state in $S$.
 It follows that the NFA $N$ rejects the string $x_S y_S$.
 Thus the string $x_S y_S$ is in $L^c$.

 (F2) Let $S\neq T$. Then without loss of generality,
 there is  a state $i$ such that $i\in S$ and $i\notin T$.
 By $x_S$, the initial state $1$ goes to $S$, so it also goes to the state $i$.
 Since $i\notin T$,
 the string $x_T$ is accepted by $N$ from $i$.
 Therefore, the NFA $N$ accepts the string $x_S y_T$,
 and so this string is not in $L^c$.

 Hence $\mathcal{F}$ is a fooling set for $L^c$ of size $2^n$.
 By Lemma~\ref{le:fool}, we have  $\nsc(L^c)\ge 2^n$.
\end{proof}

\section{Intersection and Union}
\label{se:bo}

In this section, we study the incomplete and nondeterministic state complexity
of intersection and union of prefix-closed languages.
If regular languages $K$ and $L$ are accepted
by $m$-state and $n$-state NFAs, respectively,
then the language $K\cap L$
is accepted by an NFA of at most $m n$ states,
and this bound is known to be tight in the binary case \cite{hk03}.
Our first result shows that the bound $m n$ can be met
by binary prefix-closed languages.
Then, using this result, we  get the same bound on
the incomplete state complexity of intersection on prefix-closed languages.

\begin{theorem}
\label{thm:int_nsc}
 Let $K$ and $L$ be prefix-closed languages
 over an alphabet $\Sigma$ with $\nsc(K)=m$ and \mbox{$\nsc(L)=~n$.}
 Then $\nsc(K\cap L)\le m n$,
 and the bound is tight if $|\Sigma|\ge2$.
\end{theorem}

\begin{proof}
 The upper bound is the same as for regular languages \cite{hk03}.
 For tightness, consider prefix-closed binary languages
 $K = \{w\in\{a,b\}^*\mid \#_a(w) \le m-1\}$  and
 $L = \{w\in\{a,b\}^*\mid \#_b(w) \le n-1\}$
 that are accepted by an $m$-state and an $n$-state
 incomplete DFAs $A$ and $B$, respectively, shown in Figure~\ref{fig:intersection}.

 \begin{figure}[b]\label{-----fi2}
 \centerline{\includegraphics[scale=.40]{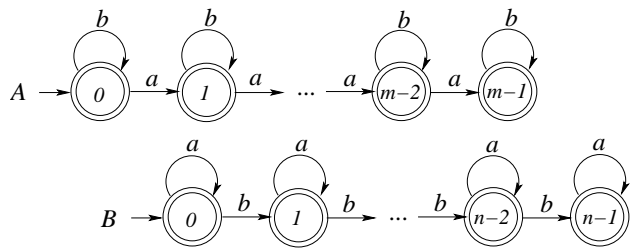}}
 \caption{The incomplete DFAs $A$ and $B$ 
              of prefix-closed languages $K$ and $L$ with $\nsc(K\cap L)=m n$.}
 \label{fig:intersection}
 \end{figure}

 Consider the set of pairs of strings
 $\mathcal{F}=\{(a^i b^j, a^{m-1-i} b^{n-1-j})
 \mid 0\le i \le m-1, 0 \le j \le n-1\}$
 of size $m n$.
 Let us show that $\mathcal{F}$
 is a fooling set for the language $K\cap L$.

 (F1) The string $a^i b^j \cdot a^{m-1-i} b^{n-1-j}$
 has exactly $m-1$ $a$'s and $n-1$ $b$'s.
 It follows that it is in $K\cap L$.

 (F2) Let $(i,j)\neq(k,\ell)$.
 If $i<k$, then the string $a^k b^{\ell}\cdot  a^{m-1-i} b^{n-1-j}$
 contains $m-1+(k-i)$ $a$'s,
 and therefore it is not in $K$.
 The case of $j<\ell$ is symmetric.

 Hence $\mathcal{F}$ is a fooling set for $K\cap L$,
 and the theorem follows.
\end{proof}

\begin{theorem}
 Let $K$ and $L$ be prefix-closed languages
 over an alphabet $\Sigma$ with $\isc(K)=m$ and \mbox{$\isc(L)=~n$.}
 Then $\isc(K\cap L)\le m n$,
 and the bound is tight if $|\Sigma|\ge2$.
\end{theorem}

\begin{proof}
 Let $A=(Q_A,\Sigma,\delta_A,s_A,Q_A)$ and
 $B=(Q_B,\Sigma,\delta_B,s_B,Q_B)$
 be incomplete DFAs for  $K$ and $L$, respectively.
 Define an incomplete product automaton
 $M=(Q_A\times Q_B, \Sigma, \delta, (s_A,s_B), Q_A\times Q_B)$,
 where 
 $$
     \delta((p,q),a) = \begin{cases}
            (\delta_A(p,a),\delta_B(q,a)),
                 & \text{if both $\delta_A(p,a)$  and $\delta_B(q,a)$ are defined,}\\
          \text{undefined},  &\text{otherwise}.
   \end{cases}
 $$
 The  DFA $M$ accepts the language $K\cap L$.
 This gives the upper bound $m n$.
 For tightness, consider the same languages $K$ and $L$
 as in the proof of the previous theorem.
 Notice that $K$ and $L$ are accepted by  $m$-state and $n$-state
 incomplete DFAs, respectively.
 We have shown that nondeterministic state complexity of their intersection
 is $m n$. It follows that the incomplete state complexity is also at least $m n$.
\end{proof}

Our next result
on the incomplete state complexity 
of union on prefix-closed languages
can be derived from the result on the quotient complexity of union in \cite{bjz14}.
For the sake of completeness, we restate it in terms of incomplete complexities,
and recall the proof.

\begin{theorem}
  Let $K$ and $L$ be prefix-closed languages
 over an alphabet $\Sigma$ with $\isc(K)=m$ and \mbox{$\isc(L)=~n$.}
 Then $\isc(K\cup L)\le m n + m +n$,
 and the bound is tight if $|\Sigma|\ge2$.
\end{theorem}

\begin{proof}
 Let $A=(\{0,1,\ldots,m-1\},\Sigma,\delta_A,0,F_A)$ and 
 $B=(\{0,1,\ldots,n-1\},\Sigma,\delta_B,0,F_B)$
 be incomplete DFAs for the languages $K$ and $L$, respectively.
 To construct a DFA for the language $K\cup L$,
 we first add the dead states $m$ and $n$ to the DFAs $A$ and $B$,
 and let go all the undefined transitions to the dead states.
 Now we construct the classic product-automaton
 from the resulting complete DFAs
 with the state set
 $\{0,1,\ldots,m\}\times\{0,1,\ldots,n\}$.
 All its states are final, except for the state $(m,n)$
 that is dead, and we do not count it.
 Hence we get the upper bound $m n + m +n$
 on the incomplete state complexity of union.

 \begin{figure}[h!]\label{-----fi3}
  \centerline{\includegraphics[scale=.40]{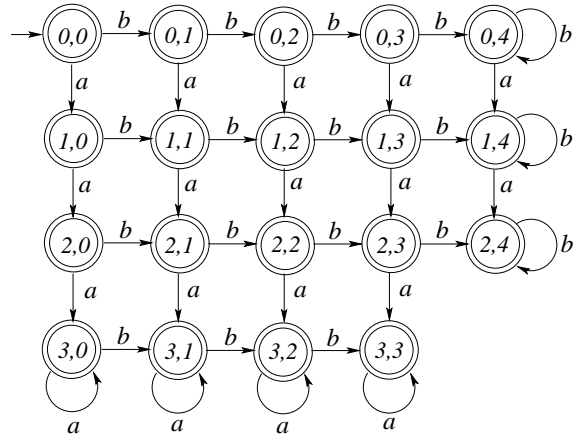}}
  \caption{The product automaton for incomplete DFAs $A$ and $B$
        from Figure~\ref{fig:intersection}; $m=3$ and $n=4$.}
  \label{fig:union_product}
  \end{figure}

 For tightness, we again consider the languages
 described in the proof of Theorem~\ref{thm:int_nsc}.
 We add the dead states $m$ and $n$ and construct the product automaton.
 The product automaton in the case of $m=3$ and $n=4$
 is shown in Figure~\ref{fig:union_product}.
 
 Each state $(i,j)$ of the product automaton is reached 
 from the initial state $(0,0)$ by the string $a^ i b^j$.
 Let $(i,j)$ and $(k,\ell)$ be two distinct states of the product automaton.
 If $i<k$, then the string $a^{m-k} b^n$ is rejected from  $(k,\ell)$
 and accepted from $(i,j)$.
 If $j<\ell$, then the string $b^{n-\ell} a^m$
 is rejected from $(k,\ell)$
 and accepted from $(i,j)$.
 Thus all the states in the product-automaton
 are reachable and pairwise distinguishable,
 and the lower bound $mn+m+n$ follows.
\end{proof}

In the nondeterministic case,
the upper bound for union on regular language is $m+n+1$,
and it is tight in the binary case \cite{hk03}.
We get the same bound for union on prefix-closed languages,
however, to define witness languages,
we use a four-letter alphabet.

\begin{theorem}
 Let $K$ and $L$ be prefix-closed languages
 over an alphabet $\Sigma$ with $\nsc(K)=m$ and \mbox{$\nsc(L)=n$.}
 Then $\nsc(K\cup L)\le m+ n+1$,
 and the bound is tight if $|\Sigma|\ge4$.
\end{theorem}

\begin{proof}
 \begin{figure}[t]\label{-----fi4}
 \centerline{\includegraphics[scale=.40]{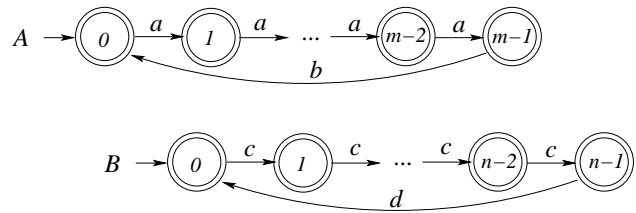}}
 \caption{The NFAs $A$ and $B$ of prefix-closed languages $K$ and $L$ with
       $\nsc(K\cup L)=m+n+1$.}
 \label{fig:union}
 \end{figure}
 The upper bound is the same as for regular languages \cite{hk03}.
 To prove tightness, let $K$ and $L$ be the prefix-closed languages 
 accepted by the NFAs $A$ and $B$, respectively,
 shown in Figure~\ref{fig:union}.
 Let
 \begin{align*}
   &\mathcal{A}= \{(a^i , a^{m-1-i} b)
              \mid i=1,2,\ldots,m-1\}\cup\{(a^{m-1}b,a)\},\\
   &  \mathcal{B}= \{(c^j , c^{n-1-j} d)
              \mid j=1,2,\ldots,n-1\}\cup\{(c^{n-1}d,c)\}.
 \end{align*}
 Let us show that $\mathcal{A} \cup \mathcal{B}$
 is a fooling set for the language $K\cup L$.

 (F1) We have $a^i \cdot a^{m-1-i} b = a^{m-1}b$
 and $c^j \cdot c^{n-1-j} d = c^{n-1} d$. Both these strings are in $K\cup L$.
 The strings $a^{m-1} b\cdot a$ 
 and $c^{n-1} d \cdot c$ are in $K\cup L$ as well.

 (F2) If $1\le i < i' \le m-1$, then the string $a^i \cdot a^{m-1-i'} b$
 is not in $K$ since $ m-1-(i'-i)<m-1$.
 Next, if $1 \le i\le m-1$, then $a^{m-1} b\cdot a^{m-1-i} b$  is not in $K$.
 The argumentation for two pairs from $\mathcal{B}$
 is similar.
 If~we concatenate the first part of a pair in  $\mathcal{A}$
 with the second part of a pair in $\mathcal{B}$,
 then we get a string that either contains all three symbols $a,c,d$,
 or contains both symbols $a$ and $d$.
 No such string is in $K\cup L$.

 Thus $\mathcal{A} \cup \mathcal{B}$
 is a fooling set for the language $K\cup L$.
 Moreover, the sets $\mathcal{A}\cup \{(\eps,c)\}$
 and  $\mathcal{B}\cup \{(\eps,a)\}$
 are fooling sets for $K\cup L$ as well.
 By Lemma~\ref{le:fool1},
 we have $\nsc(K\cup L)\ge m+n+1$.
\end{proof}

\section{Concatenation}
\label{se:cs}

In this section, we deal with the concatenation operation
on prefix-closed languages. 
We start with incomplete state complexity.
We use a slightly different ternary witness language than in \cite{bjz14},
and prove the upper bound using automata constructions.

\begin{theorem}
 Let $m,n\ge3$.
 Let $K$ and $L$ be  prefix-closed languages 
 over an alphabet $\Sigma$ with $\isc(K)=m$ and $\isc(L)=n$.
 Then $\isc(KL)\le m\cdot 2^{n-1}+2^n-1$,
 and the bound is tight if $|\Sigma|\ge3$.
\end{theorem}

\begin{proof}
 Let 
 $A=(Q_A,\Sigma,\delta_A,s_A,Q_A)$ and
 $B=(Q_B,\Sigma,\delta_B,s_B,Q_B)$ be incomplete DFAs 
 with all states final
 accepting the languages $K$ and $L$,
 respectively.
 Construct an NFA $N$ for the language $KL$
 from the DFAs $A$ and $B$
 by adding the transition on a symbol $a$ 
 from a state $q$ in $Q_A$ to the initial state $s_B$ of $B$
 whenever the transition on $a$ in state $q$ is defined in $A$.
 The initial states of the NFA $N$ are $s_A$ and $s_B$,
 and the set of final states is $Q_B$.
 Each reachable subset of the subset automaton
 of the NFA $N$ contains at most one state of $Q_A$,
 and several states of $Q_B$.
 Moreover, if a state of $Q_A$ is in a reachable subset $S$,
 then $S$ must contain the state $s_B$.
 This gives the upper bound $m\cdot 2^{n-1}+2^n-1$
 on $\isc(KL)$ since the empty set is not counted.

 For tightness,
 consider the prefix-closed languages $K$ and $L$
 accepted by incomplete DFAs  $A$ and $B$, respectively,
 shown in Figure~\ref{fig:concatenation},
 in which the transitions are as follows:

 on $a$, state $q_0$ goes to itself, 
 and each state $j$ goes to $(j+1)\bmod n$;
     
 on $b$, each state $q_i$ goes to state $q_0$, state $0$ goes to itself,
 and state $j$ with $1\le j\le n-2$ goes to $j+1$;

 on $c$, each state $q_i$ with $0 \le i \le m-2$ goes to $q_{i+1}$,
 and each state $j$ goes to itself;\\
\noindent
 and all the remaining transitions are undefined.

 Construct an NFA $N$ for the language $KL$ as described above.
 Let us show that the subset automaton of the NFA $N$
 has $m\cdot 2^{n-1}+2^n-1$
 reachable and pairwise distinguishable non-empty subsets.

 \begin{figure}[t]\label{-----fi5}
  \centerline{\includegraphics[scale=.40]{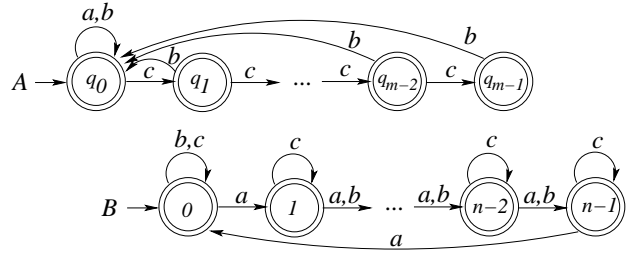}}
  \caption{The incomplete DFAs $A$ and $B$ of languages $K$ and $L$
      with $\isc(KL)=m\cdot2^{n-1}+2^n-1$.}
  \label{fig:concatenation}
  \end{figure}

 (1) First, let us show that each set $\{q_0\}\cup S$ is reachable, where
 $S\subseteq\{0,1,\ldots,n-1\}$ and $0\in S$. 
 The proof is by induction on the size of subsets.
 The set $\{q_0,0\}$ is the initial subset.
 The set $\{q_0,0, j_1,j_2,\ldots,j_k\}$ with 
 $1\le j_1<j_2<\cdots  <j_k\le n-1$ is reached from the set
 $\{q_0,0,j_2-j_1,\ldots,j_k-j_1\}$ by the string $a b^{j_1-1}$,
 and the latter set is reachable by induction.

 (2) Now, let us show that each set $\{q_i\}\cup S$,
 is reachable, where
 $1\le i \le m-1$, $S\subseteq\{0,1,\ldots,n-1\}$ and $0\in S$.
 The set $\{q_i\}\cup S$ is reached from $\{q_0\}\cup S$ by $c^i$,
 and the latter set is reachable as shown in  (1).

 (3) Next, we show that each set $S$ with
 $S\subseteq\{0,1,\ldots,n-1\}$ and $0\in S$ is reachable.
 The set $S$ is reached from $\{q_{m-1}\}\cup S$ by $c$,
 and the latter set is reachable as shown in case (2).

 (4) Finally, we show that each non-empty set $S$ with
 $S\subseteq\{0,1,\ldots,n-1\}$ and $0\notin S$ is reachable.
 If $S=\{j_1,j_2,\ldots,j_k\}$ with $j_1\ge1$,
 then $S$ is reached from the set $\{0,j_2-j_1,\ldots,j_k-j_1\}$ by $a^{j_1}$,
 and the latter set is reachable as shown in case (3).

 This proves the  reachability of $m\cdot 2^{n-1}+2^n-1$ non-empty subsets.

 To prove distinguishability,
 notice that the string $b^n$ is accepted by the DFA $B$
 only from the state 0,
 and the string $a^{n-1-i}a b^n$ is accepted only from the state $i$
 ($1\le i\le n-1$).
 If $S$ and $T$ are two distinct subsets of $\{0,1,\ldots,n-1\}$,
 then $S$ and $T$
 differ in a state $i$.
 If $i=0$, then $b^n$ distinguishes $S$ and $T$,
 and if $i\ge 1$, then $a^{n-i}b^n$ distinguishes $S$ and $T$.

 Next, the sets $\{q_i\}\cup S$ and $\{q_i\}\cup T$,
 where $S$ and $T$ are distinct subsets of $\{0,1,\ldots,n-1\}$,
 go to $S$ and $T$, respectively, by $c^m$.
 Since  $S$ and $T$ are distinguishable,
 the sets $\{q_i\}\cup S$ and $\{q_i\}\cup T$
 are distinguishable as well.

 Finally, notice that the string $b^n a b^n$
 is accepted by the NFA $N$ from each state $q_i$,
 but rejected from each state $i$ in $\{0,1,\ldots,n-1\}$.
 Hence the sets $\{q_i\}\cup S$ and $T$,
 where $S$ and $T$ are subsets of $\{0,\ldots,n-1\}$,
 are distinguishable.
 Now let $0\le i <j\le m-1$.
 Then $\{q_i\}\cup S$ and $\{q_j\}\cup T$ 
 go  to 
 $\{q_{i+m-j}\}\cup S$ and $T$, respectively,  by $c^{m-j}$.
 Since  $\{q_{i+m-j}\}\cup S$ and $T$ are distinguishable,
 the sets $\{q_i\}\cup S$ and $\{q_j\}\cup T$
 are distinguishable as well.
 This proves the distinguishability of all the reachable subsets,
 and completes the proof.
\end{proof}

In the next theorem, we consider the nondeterministic case.
For regular languages, the upper bound
on the nondeterministic state complexity of concatenation 
is $m+n$, and it is tight in the binary case \cite{hk03}.
For prefix-closed languages,
we get the same bound for concatenation.
However, we define witness languages over a ternary alphabet.

\begin{theorem}
 Let $m,n\ge3$.
 Let $K$ and $L$ be  prefix-closed languages 
 over an alphabet $\Sigma$ with $\nsc(K)=m$ and $\nsc(L)=n$.
 Then  $\nsc(KL)\le m+n$,
 and the bound is tight if $|\Sigma|\ge3$.
\end{theorem}

\begin{proof}
 The upper bound is the same as for regular languages \cite{hk03}.
 For tightness, consider the ternary prefix-closed languages $K$ and $L$
 accepted by incomplete DFAs $A$ and $B$, respectively,
 shown in  Figure~\ref{fig:productNSC}.
 Notice that if a  string $w$ is  in $KL$, then $w$ is in the language 
 $b^* a^* c^* b^* a^* c^*$, and the number of $a$'s in $w$
 is at most $(n+m-2)$.

 \begin{figure}[t]\label{-----fi6}
 \centerline{\includegraphics[scale=.40]{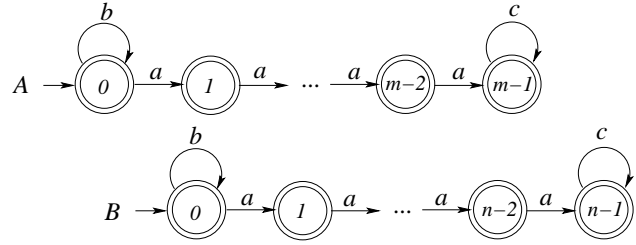}}
 \caption{The incomplete DFAs of prefix-closed languages $K$ and $L$ 
              with $\nsc(KL)=m+n$.}
 \label{fig:productNSC}
 \end{figure}

 For $i=0,1,\ldots,m+n-1$, define the pair $(x_i, y_i)$ as follows:
 \begin{align*}
      (x_i, y_i) &= ( a^i, a^{m-1-i} c b a^{n-1} ),  
     \quad\text{for $i=0,1,\ldots m-1$,}\\
    (x_{m+j}, y_{m+j} ) &= ( a^{m-1} c b a^j , a^{n-1-j} ), 
                      \quad \text{for $j=0,1,\ldots n-1$.}
 \end{align*}
 Let us show that the set of pairs
 $\mathcal{F} = \{ (x_i, y_i)\mid i=0,1,\ldots,m+n-1\}$
 is a fooling set for the language~$KL$.

 (F1) For each $i$, we have $x_i y_i = a^{m-1} c b a^{n-1}$.
 Thus $x_i y_i$ is in $KL$ since $a^{m-1} c $ is in $K$ and $b a^{n-1}$ is in $L$.

 (F2) Let $i<j$ and $(i,j)\neq(m-1,m)$.
 Then the number of $a$'s in the string $x_j y_i$ is greater than  $m+n-2$,
 and therefore the string $x_j y_i$ is not in $KL$.
 If  $(i,j)=(m-1,m)$, then $x_m y_{m-1} = a^{m-1} c b c b a^{n-1}$.
 Thus $x_m y_{m-1}$ is not in $b^* a^* c^* b^* a^* c^*$,
 and therefore it is not in $KL$.

 Hence the set $\mathcal{F}$ is a fooling set for the language $KL$,
 so $\nsc(KL)\ge m+n$.
\end{proof}

\section{Star and Reversal}
\label{se:sr}

We conclude our paper with the star and reversal operation 
on prefix-closed languages.
The star of a language $L$ is the language $L^*=\bigcup_{i\ge0} L^i$,
where $L^0=\{\eps\}$ and $L^{i+1}=L^i\cdot L$.

If a regular language  $L$ is accepted by a complete $n$-state DFA,
then the language $L^*$ is accepted by a DFA of at most
$3/4\cdot 2^n$ states, and the bound is tight in the binary case 
\cite{ma70,yzs94}. 

For prefix-closed languages, 
the upper bound on the quotient complexity for star is $2^{n-2}+1$,
and it has been shown to be tight in the ternary case \cite{bjz14}.
In the case of incomplete state complexity,
we get the bound $2^{n-1}$.
For the sake of completeness, we give a simple proof of the upper bound
using automata constructions.
Moreover, we are able to define a witness language over a binary alphabet.

\begin{theorem}
 Let $n\ge4$.
 Let $L$ be a prefix-closed regular language 
 over an alphabet $\Sigma$ with $\isc(L)=n$.
 Then $\isc(L^*)\le2^{n-1}$,
 and the bound is tight if $|\Sigma|\ge2$.
\end{theorem}

\begin{proof}
  Let $A=(Q,\Sigma,\cdot,s,Q)$
 be an incomplete DFA for $L$.
 Construct an NFA $A^*$ for $L^*$ from the DFA $A$
 by adding the transition on a symbol $a$
 from a state $q$ to the initial state $s$
 whenever the transition $q\cdot a$ is defined.
 In the subset automaton of the NFA $A^*$,
 each reachable set is either empty,
 or it contains the initial state $s$.
 It follows that $\isc(L^*)\le2^{n-1}$.

 For tightness, consider
 the binary incomplete DFA 
 with the state set $\{1,2,\ldots,n\}$,
 the initial state $1$ and with all states final.
 The transitions are as follows.
 By $a$, the transitions in states 1 and 2 are undefined,
 each odd state $i$ with $3\le i \le n-1$ goes to $i+1$, and
 each even state $i$ with $3\le i \le n-1$ goes to $i-1$.
 By $b$, there is a cycle $(1,2,3)$,
 each odd state $i$ with $4\le i \le n-1$ goes to $i-1$, and
 each even state $i$ with $4\le i \le n-1$ goes to $i+1$.
 If $n$ is odd, then $n$ goes to itself by $a$,
 otherwise it goes to itself by~$b$.
 The DFA for $n=6$ is shown in Figure~\ref{fig:star}.

 \begin{figure}[t]\label{-----fi7}
  \centerline{\includegraphics[scale=.40]{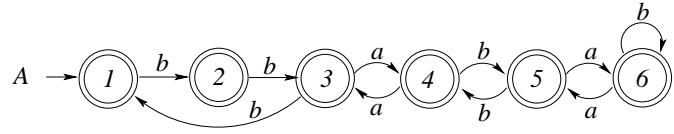}}
  \caption{The incomplete DFA $A$ of a prefix-closed language $L$
           with $\isc(L^*)=2^{n-1}$; $n=6$.}
  \label{fig:star}
  \end{figure}

 Notice that each state $i$ with $3\le i \le n$
 has exactly one in-transition on $a$ and on $b$.
 Denote by $a^{-1}(i)$ the state that goes to $i$ on $a$,
 and by $b^{-1}(i)$ the state that goes to $i$ on $b$.

 Construct an NFA $A^*$ as described above.
 Let us show that in the subset automaton of the NFA $A^*$,
 all subsets of $\{1,2,\ldots,n\}$ containing state $1$
 are reachable and pairwise distinguishable.

 We prove reachability by induction on the size of subsets.
 The basis is $|S|=1$,
 and the set $\{1\}$ is reachable since it is the initial state of 
 the subset automaton.
Assume that every set $S$ containing $1$
with $|S| = k,$ where  $1\leqslant{k}\leqslant{n-1},$ is reachable.
Let $S=\{1,i_1,i_2,i_3,\ldots,i_k\},$ where $2\leqslant{i_1}<{i_2}<{\cdots}<{i_k}\leqslant{n},$
be a set of size $k+1$.
Consider three cases:
\begin{enumerate}
     \item [(\textit{i})]$ i_{1}= 2.$ 
Take $S^{'} = \{1,b^{-1}(i_2),b^{-1}(i_3),\ldots,b^{-1}(i_k)\} $. 
Then $|S^{'}| = k $, and therefore $S^{'}$ is reachable by the induction hypothesis.
Since we have
 $
          S'\xrightarrow{b}\{1,2, i_2,\ldots,i_k\} = S,
$
 the set $S$ is reachable.

 \item[(\textit{ii})]$ i_{1}= 3.$
Take $S'=\{1,2,b^{-1}(i_2),b^{-1}(i_3),\ldots,b^{-1}(i_k)\}$. 
Then $|S'|=k+1$ and $S'$ contains states $1$ and $2$. Therefore, the set $S'$ is reachable as shown in case $(i)$.
Since we have
$
     S^{'}\xrightarrow{b}\{1,2,3,i_2,i_3,\ldots,i_k\}
              \xrightarrow{aa}\{1,3,i_2,i_3,\ldots,i_k\}=S,
$
the set $S$ is reachable.

 \item[(\textit{iii})] Let $i_1=j\ge3$,
 and assume that each set $\{1,j,i_2,\ldots,i_k\}$ is reachable.
 Let us show that then also each set $\{1,j+1,i_2,\ldots,i_k\}$
 is reachable.
 If $j$ is odd, then the set
 $\{1,j+1,i_2,\ldots,i_k\}$ is reached from the set
 $\{1,j,a^{-1}(i_2),a^{-1}(i_3),\ldots,a^{-1}(i_k)\} $ by $a$.
 If $j$ is even, then  the set
 $\{1,j+1,i_2,\ldots,i_k\}$ is reached from the set
 $\{1,j,b^{-1}(i_2),b^{-1}(i_3),\ldots,b^{-1}(i_k)\} $ by $b a a$.
\end{enumerate}

This proves reachability.
To prove distinguishability,
 notice that the string $(ab)^{n-2}$
 is accepted by the NFA $A^*$ from state $3$ since
 state $3$ goes to the initial state $1$ by $(ab)^{n-2}$
 through the computation 
  $$3\xrightarrow{ab}5\xrightarrow{ab}7\xrightarrow{ab}
   \cdots\xrightarrow{ab} n \xrightarrow{a} n\xrightarrow{b} n-1
\xrightarrow{ab} n-3\xrightarrow{ab}
   \cdots\xrightarrow{ab} 4\xrightarrow{ab}1 $$
if $n$ is odd, and through a similar computation if $n$ is even.
 On the other hand, the  string $(ab)^{n-2}$ 
cannot be read from any  other state $2i$ with $2\le i\le n/2$
 since we have
$$
    2i \xrightarrow{ab} \{2i-2,1,2\}
        \xrightarrow{ab} \{2i-4,1,2\}
 \xrightarrow{ab} \cdots  \xrightarrow{ab}\{4,1,2\} \xrightarrow{a}\{3,1\}
 \xrightarrow{b}\{1,2\} \xrightarrow{ab} \emptyset,
$$
 thus $2i$ goes to the empty set by $(ab)^i$, so also by $(ab)^{n-2}$.
 If $n$ is odd, then we have
$$
      2i+1  \xrightarrow{ab} \{2i+3,1,2\}
 \xrightarrow{ab} \{2i+5,1,2\}
 \xrightarrow{ab} \cdots  \xrightarrow{ab} \{n,1,2\}
\xrightarrow{a} \{n,1\} \xrightarrow{b}\{n-1,1,2\} 
\xrightarrow{ab}$$ $$  \{n-3,1,2\}\xrightarrow{ab} \cdots \xrightarrow{ab} 
\{2i,1,2\} \xrightarrow{(ab)^i} \emptyset,
$$
thus $2i+1$ goes to the empty set by $(ab)^{n-i}$, $i\ge2$,
and so also by $(ab)^{n-2}$.
For $n$ even, the argument is similar.
The string $(ab)^{n-2}$ is not accepted from states 1 and 2.
Hence the NFA $A^*$ accepts the string $(ab)^{n-2}$
only from the state 3.
 Since there is exactly one in-transition on $b$ in state $3$,
 and it goes from state $2$,
 the string $b (ab)^{n-2}$
 is accepted by $A^*$ only from state $2$.
 Similarly, the string $ b b (ab)^{n-2}$
 is accepted
 by $A^*$ only from state $1$.
 Next, for similar reasons, 
 the string $a (ab)^{n-2}$ is accepted only from $4$,
 the string $ b a (ab)^{n-2}$ is accepted only from $5$,
 and in the general case, the string
 $(ab)^i a (ab)^{n-2}$ is accepted only from $4+2i$  ($i\ge0$),
 and the string
 $(ba)^i (ab)^{n-2}$ is accepted only from $3+2i$  ($i\ge1$).
 Hence for each state $q$ of the NFA $A^*$,
 there exists a string $w_q$ that is accepted by $A^*$
 only from the state $q$.
 It follows that all the subsets of the
 subset automaton of the NFA $A^*$
 are pairwise distinguishable
 since two distinct subsets differ in a state $q$,
 and the string $w_q$ distinguishes the two subsets.
 This completes the proof.
\end{proof}

We did some computations in the binary case.
Having the files of $n$-state minimal binary pairwise non-isomorphic complete DFAs with a dead state and all the remaining states final,
we computed the state complexity of the star of languages accepted
by DFAs on the lists; 
here the state complexity of a regular language $L$, $\sc(L)$,
is defined as the smallest number of states
in any \emph{complete} DFA for the language~$L$.
We computed the frequencies of the resulting complexities,
and the average complexity of star.
Our results are summarized in Table~\ref{ta}.
Notice that for $n=3,4,5$,
there is just one language with $\sc(L)=n$ and $\sc(L^*)=2$.
Let us show that this holds for every $n$ with $n\ge3$.

\begin{table}[b]
\begin{center}
\setlength{\extrarowheight}{4pt}
\begin{tabular}{|c|c|c|c|c|c|c|c|c|c|c|}
\hline
$n \backslash \sc(L^*) $      & \ \ \ 1\  \ \ &\ \ \  2\  \ \   &\  \ \   3 \ \ \  
 &\ \ \    4 \ \ \   &  \ \ \  5  \ \ \  &  \ \ \  6  \ \ \  & \ \ \   7  \ \ \ 
&\ \ \ 8 \ \ \ &\ \ \ 9 \ \ \ & average\\
\hline
2       &-  &2 & - & -  & - & -  & -  &-&-&  2\\
\hline
3         & 8  & 1  & 6 & -  & - & -  & -  &  -&-&  1.866\\
\hline 
4        &161 &1 &48& 30& 6& -  & -  &  -&-&  1.857\\
\hline
5        &4177 &1&771& 275& 350&84&84& -&26&   1.849\\
\hline
\end{tabular}
\vskip10pt
\caption{The frequencies of the complexities
     and the average complexity of star
 on prefix-closed languages in the binary case; $n=2,3,4,5.$}
\label{ta}
\end{center}
\end{table}

\begin{proposition}
 Let $n\ge3$.
 There exists exactly one (up to renaming of alphabet symbols)
 binary prefix-closed regular language $L$
 with $\sc(L)=~n$ and $\sc(L^*)=2$.
\end{proposition}

\begin{proof}
 Let $A=(\{0,1\}, \{a,b\}, \delta, 0, F)$
 be a minimal two-state DFA for the language $L^*$.
 Since $L$ is prefix-closed, 
 the language $L^*$ is prefix-closed as well.
 It follows that state 0 is final, and state 1 is dead,
 thus $F=\{0\}$ and $\delta(1,a)=\delta(1,b)=1$.

 Without loss of generality,
 state 1 is reached from the initial state 0 by $a$,
 thus $\delta(0,a) =1$.

 Since $n\ge3$, the language $L$ contains a non-empty string.
 This means that the language $L^*$  contains a non-empty string as well.
 Therefore, we must have $\delta(0,b)=0$,
 and so $L^*=b^*$.

 \begin{figure}[t]\label{-----fi8}
 \centerline{\includegraphics[scale=.40]{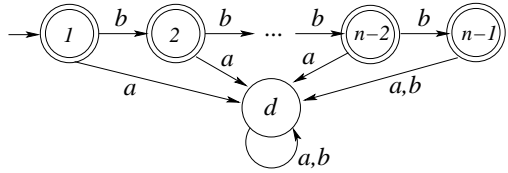}}
 \caption{The only binary $n$-state complete DFA of 
              a prefix-closed language $L$ with $\sc(L^*)=2$.}
 \label{fig:allFinal2}
 \end{figure}

 Now let $B$ be the minimal $n$-state  DFA for $L$.
 Then all the states of $B$ are final, except for  the dead state.
 Since $L^*=b^*$,
 no $a$ may occur in any string of $L$.
 Hence each non-dead state of $B$
 must go to the dead state on $a$.
 Since all states must be reachable,
 we must have a path labeled by $b^{n-2}$ and
 going through all the final states.
 The last final state must go to the dead state on $b$
 because otherwise all  final states
 would be equivalent.
 The resulting $n$-state DFA $B$
 is shown in Figure~\ref{fig:allFinal2}. 
\end{proof}

The reverse $w^R$ of a string $w$ is defined by $\eps^R=\eps$,
and $(w a)^R= a w^R$ for $a$ in $\Sigma$ and $w$ in $\Sigma^*$.
The reverse of a language $L$ is the language 
$L^R=\{w^R\mid w\in L\}$.
If a regular language $L$ is accepted by a complete $n$-state DFA,
then the language $L^R$ is accepted by  a complete DFA of at most $2^n$
states \cite{mi66,yzs94},
and the bound is tight in the binary case \cite{js12,le81}.

For prefix-closed languages, the quotient complexity of reversal
is $2^{n-1}$  \cite{bjz14}, and it
follows from the results on ideal languages \cite{bjl13}
since reversal commutes with complementation,
and the complement of a prefix-closed language is a right ideal;
here a language $L$ is a right ideal if $L=L\cdot \Sigma^*$.

We restate the result for reversal in terms of incomplete state complexity,
and prove tightness using a slightly different witness language.

\begin{theorem}
 Let $n\ge2$.
 Let $L$ be a prefix-closed regular language 
 over an alphabet $\Sigma$ with $\isc(L)=n$.
 Then $\isc(L^R)\le2^{n}-1$,
 and the bound is tight if $|\Sigma|\ge2$.
\end{theorem}

\begin{proof}
 Let $A$ be an incomplete DFA for $L$.
 Construct an NFA $A^R$ for the language $L^R$
 from the DFA $A$
 by swapping the role of the initial and final states,
 and by reversing all the transitions.
 The subset automaton of the NFA $A^R$ 
 has at most $2^n-1$ non-empty reachable states,
 and the upper bound follows.

 For tightness,
 consider the incomplete DFA $A$ with all states final,
 shown in Figure~\ref{fig:reversal}.
 Construct an NFA $A^R$ as described above.
 In the subset automaton of the NFA $A^R$,
 the initial state is $\{1,2,\ldots,n\}$.
 If $S$ is a subset and if $i\in S$,
 then the subset $S\setminus\{i\}$
 is reached from $S$ by $a^ i b a^{n-i}$.
 This proves the reachability of all non-empty
 subset by odd induction.
 Since the states of the subset automaton of
 any reversed DFA are pairwise distinguishable \cite{ckp,js12,mi66},
 the theorem follows.

 \begin{figure}[t]\label{-----fi9}
 \centerline{\includegraphics[scale=.35]{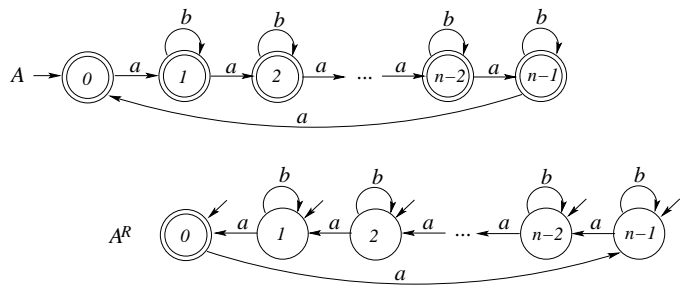}}
 \caption{The incomplete DFA $A$ of a  language $L$
               with $\isc(L^R)=2^n-1$, and the NFA $A^R$. }
 \label{fig:reversal}
 \end{figure}
 \end{proof}

Now, let us turn to the nondeterministic case.
For regular languages, the tight bound for both star and reversal is
$n+1$. It is met by a unary language for star \cite{hk03},
and by a binary language for reversal \cite{ji05}.

For prefix-closed languages, we get the same bound for reversal.
However, for star, the upper bound is $n$
since every prefix-closed language contains the empty string,
and there is no need to add a new initial state in the 
 construction of an NFA for star.
In the following theorem, 
we show that both these bounds are tight in the binary case.

\begin{theorem}
 Let $n\ge2$.
 Let $L$ be a prefix-closed language 
 over an alphabet $\Sigma$ with $\nsc(L)=n$.
 Then \\
 \hglue10pt (1) $\nsc(L^*)\le n$,\\
 \hglue10pt (2) $\nsc(L^R)\le n+1$,\\
 and both bounds are tight if $|\Sigma|\ge2$.
\end{theorem}

\begin{proof}
 (1) Let $N=(Q,\Sigma,\delta,s, F)$ be an $n$-state NFA for $L$.
 Since $L$ is prefix-closed, the empty string is in~$L$.
 Therefore, we can get an $n$-state NFA for the language $L^*$
 from the NFA $N$  as follows:
 for each state $q$ and each symbol $a$ such that
 $\delta(q,a)\cap F\neq \emptyset$,
 we add a transition on $a$ from $q$ to the initial state $s$.
 Thus the upper bound is $n$.

 For tightness, consider the prefix-closed language $L$
 accepted by the NFA $A$ shown in Figure~\ref{fig:star_nsc}.
 Consider the set of pair of strings
 $\mathcal{F}=\{ (a^i, a^{n-1-i}b) \mid i=0,1,\ldots,n-1\}$
 of size $n$.
 Let us show that $\mathcal{F}$
 is a fooling set for the language~$L^*$.

 \begin{figure}[b]\label{-----fi10}
 \centerline{\includegraphics[scale=.40]{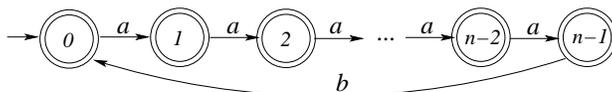}}
 \caption{The NFA of a prefix-closed language $L$ with $\nsc(L^*)=n$
              and $\nsc(L^R)=n+1$.}
 \label{fig:star_nsc}
 \end{figure}

 (F1) We have $a^i  a^{n-1-i}b=a^{n-1}b$.
 Since the string $a^{n-1}b$ is in $L$,
 it also is in $L^*$.

 (F2) Let $i<j$. Then $a^i a^{n-1-j}b = a^{n-1-(j-i)} b$.
 Since no string $a^\ell b$ with $\ell<n-1$ is in $L$,
 the string $ a^{n-1-(j-i)} b$ is not in $L^*$.

 Hence the set $\mathcal{F}$
 is a fooling set for the language $L^*$,
 and the lower bound follows.

\smallskip
 (2) The upper bound is the same as 
 for regular languages \cite{hk03}.
 It is shown in \cite[Theorem~2]{ji05}
 that this bound is met by the binary prefix-closed language $L$
 accepted by the NFA   shown in Figure~\ref{fig:star_nsc}.
 The~proof~in \cite{ji05} is by a counting argument.
 Notice that Lemma~\ref{le:fool1} is satisfied for the language $L^R$ with
 $
 \mathcal{A}=\{(b a^i, a^{n-1-i})\mid i=0,1,\ldots,n-2\},
 \mathcal{B}=\{(b a^{n-1}, b a^{n-1})\},
 u = b a^{n-1}$, and
$ v = a$. This gives $\nsc(L^R)\ge n+1$ immediately.
\end{proof}

\section{Conclusions}

In this paper we considered operations
on languages recognized by incomplete deterministic
or nondeterministic finite automata with all states final.
Our results are summarized in Tables~\ref{ta} and~\ref{ta2}.
The results on quotient (state) complexity 
on prefix-closed languages are from \cite{bjz14},
and the results for regular languages are from \cite{hk03,ji05,ma70,yzs94}.
Notice that in the nondeterministic case,
our results are the same as in the general case of regular languages,
except for the star operation.
However, to prove tightness,
we usually  used larger alphabets than in the general case.
Whether or not these bounds are tight also for smaller alphabets remains open.

\begin{table}[t]
\begin{center}
\begin{tabular}{|ll|cc|cc|cc|}
\hline
&& complement &$|\Sigma|$&intersection&$|\Sigma|$& union & $|\Sigma|$    \\ 
\hline
$\isc$     &on prefix-closed   
& $n+1$ & 1 & $m n$ & 2& $m n+m+n$ &2  \\
\hline
$\sc$ & on prefix-closed \cite{bjz14}   
& $n$ & 1 & $m n-m-n+2$ & 2& $m n$ &2  \\
\hline
$\sc$ & on regular \cite{ma70,yzs94}    
& $n$ & 1 & $m n$ & 2& $m n$ &2  \\
\hline
\hline
$\nsc$ &on prefix-closed 
& $2^n$ & 3  & $m n$ &2& $m+n+1$ & 4\\
\hline
$\nsc$ &on regular \cite{hk03,ji05} 
& $2^n$ & 2  & $m n$ &2& $m+n+1$ & 2\\
\hline
\end{tabular}
\end{center}
\caption{The complexity
    of boolean operations on prefix-closed and regular languages.}
\label{ta}
\end{table}

\begin{table}[t]
\begin{center}
\begin{tabular}{|ll|cc|cc|cc|} 
\hline
&& concatenation &$|\Sigma|$&star&$|\Sigma|$& reversal & $|\Sigma|$ \\ 
\hline
$\isc$ &on prefix-closed  
 & $m2^{n-1}+2^n-1$ & 3& $2^{n-1}$ & 2& $2^n-1$ &2\\
\hline
$\sc$ & on prefix-closed \cite{bjz14} 
  & $(m+1)2^{n-2}$ & 3 & $2^{n-2}+1$ &3& $2^{n-1}$ &2  \\
\hline
$\sc$ & on regular \cite{ma70,yzs94}  
  & $m2^{n}-2^{n-1}$ &2 & $2^{n-1}+2^{n-2}$ & 2& $2^n$ &2  \\
\hline
\hline
$\nsc$  &on prefix-closed  
 & $m+n$ &3 &  $n$ &2 & $n+1$ &2\\
\hline
$\nsc$ &on regular \cite{hk03,ji05}
 & $m+n$ & 2  & $n+1$ &2& $n+1$ & 2\\
\hline
\end{tabular}
\end{center}
\caption{The complexity of concatenation, star, and reversal on prefix-closed
    and regular languages.}
\label{ta2}
\end{table}

 \nocite{*}
 \bibliographystyle{eptcs}
 \bibliography{allFinal}

\end{document}